\documentclass[11pt]{article}
\usepackage{amsmath,amsthm,epsfig,float,enumerate,comment,fullpage,color,graphicx}
\usepackage{hyperref}
\usepackage[T1]{fontenc}

\def\ceil#1{\lceil #1 \rceil}

\mathchardef\mhyphen="2D 
\newtheorem{theorem}{Theorem}
\newtheorem{corollary}[theorem]{Corollary}
\newtheorem{lemma}[theorem]{Lemma}

\newtheorem{definition}[theorem]{Definition}
\newtheorem{claim}[theorem]{Claim}
\newtheorem{fact}[theorem]{Fact}

\newenvironment{proof-sketch}{\noindent{\bf Sketch of Proof:}\hspace*{1em}}{\qed\bigskip}
\newenvironment{proof-idea}{\noindent{\bf Proof Idea:}\hspace*{1em}}{\qed\bigskip}
\newenvironment{proof-of-lemma}[1]{\noindent{\bf Proof of Lemma #1:}\hspace*{1em}}{\qed\bigskip}
\newenvironment{proof-of-proposition}[1]{\noindent{\bf Proof of Proposition #1:}\hspace*{1em}}{\qed\bigskip}
\newenvironment{proof-attempt}{\noindent{\bf Proof Attempt:}\hspace*{1em}}{\qed\bigskip}
\newenvironment{proofof}[1]{\noindent{\bf Proof}
of #1:\hspace*{1em}}{\qed\bigskip}

\newcommand{\eps}{\varepsilon}
\newcommand{\tr}[1]{\mathrm{tr}(#1)} 
\newcommand{\Tr}{\mathrm{Tr}} 

\newcommand{\YES}{\mbox{\sc{yes}}}
\newcommand{\NO}{\mbox{\sc{no}}}

\newcommand{\pr}{\text{Pr}}

\newcommand{\tensor}{\otimes}

\newcommand{\bra}[1]{\langle #1|}
\newcommand{\ket}[1]{|#1\rangle}
\newcommand{\braket}[2]{\langle #1|#2\rangle}
\newcommand{\ketbra}[2]{\ket{#1}{\bra{#2}}}

\newcommand{\etal}{{\it et~al.\ }}

\newcommand{\IP}{\textsf{IP}}
\newcommand{\MIP}{\textsf{MIP}}

\newcommand{\PSPACE}{\textsf{PSPACE}}

\newcommand{\NP}{\textsf{NP}}
\newcommand{\NEXP}{\textsf{NEXP}}
\newcommand{\NPC}{\textsf{NP}\mhyphen\textsf{complete}}

\newcommand{\PP}{\textsf{PP}}

\newcommand{\MA}{\textsf{MA}}

\newcommand{\QMA}{\textsf{QMA}}
\newcommand{\QMAC}{\textsf{QMA-complete}}
\newcommand{\QMAH}{\textsf{QMA-hard}}
\newcommand{\QMAT}{\textsf{QMA(2)}}
\newcommand{\BELLQMAT}{\textsf{BELL-QMA(2)}}
\newcommand{\QMATC}{\textsf{QMA(2)-complete}}
\newcommand{\QMATH}{\textsf{QMA(2)-hard}}

\newcommand{\LL}{{\cal L}}
\newcommand{\HH}{{\cal H}}

\newcommand{\GROUPNONMEMBERSHIP}{\textsc{group-non-membership}}

\newcommand{\lh}{\textsc{local }\textsc{hamiltonian}}
\newcommand{\klh}{k\mhyphen\textsc{local }\textsc{hamiltonian}}
\renewcommand{\klh}{k-\textsc{local }\textsc{hamiltonian}}
\newcommand{\sklh}{\textsc{separable }k\mhyphen\textsc{local }\textsc{hamiltonian}}
\renewcommand{\sklh}{\textsc{separable }k-\textsc{local }\textsc{hamiltonian}}
\newcommand{\sh}{\textsc{sparse } \textsc{hamiltonian}}
\newcommand{\slh} {\textsc{separable }\textsc{local }\textsc{hamiltonian}}
\newcommand{\ssh} {\textsc{separable }\textsc{sparse }\textsc{hamiltonian}}
\newcommand{\cldm}{\textsc{consistency }\textsc{of }\textsc{local }\textsc{density }\textsc{matrices}}
\newcommand{\purenrepresentability}{\textsc{pure }N-\textsc{representability}}

\newcommand{\sat}{\textsc{sat}}

\newcommand{\maxksat}{\textsc{max}\mhyphen k\mhyphen\textsc{sat}}

\newcommand{\onote}[1]{\textcolor{blue}{ {\textbf{(Or:}
#1\textbf{)}}}}
\renewcommand{\onote}[1]{}

\newcommand{\anote}[1]{\textcolor{blue} {{\textbf{(Andre:}}
#1\textbf{)}}}
\renewcommand{\anote}[1]{}

\newcommand{\fullversion}[1]{}
\renewcommand{\fullversion}[1]{#1}

\newcommand{\abstractversion}[1]{}

\definecolor{gray}{rgb}{0.5,0.5,0.5}

\newcommand{\cadre}[1]
{
\begin{tabular}{|p{15.4cm}|}
\hline
#1 \\
\hline
\end{tabular}
}
\newcommand{\spa}[1]{\mathcal{#1}}
\newcommand{\COMMENT}[1]{}

\title{The Complexity of the Separable Hamiltonian Problem}

 \author{Andr\'{e} Chailloux\thanks{LIAFA - Universit\'e Paris 7 and UC Berkeley
}
 \and
 Or Sattath
 \thanks{School of Computer Science and Engineering,
 The Hebrew University, Jerusalem, Israel.
 Supported by the Clore Fellowship program, Julia Kempe's Individual Research Grant of the Israeli Science Foundation and by Julia Kempe's European Research Council (ERC) Starting Grant.}
 }
\begin{document}
\maketitle

\begin{abstract}
In this paper, we study variants of the canonical \lh\ problem where, in addition, the witness is promised to be separable.
We define two variants of the \lh\ problem. The input for the \slh\ problem is the same as the \lh\ problem, i.e. a local Hamiltonian and two energies $a$ and $b$, but the question is somewhat different: the answer is $\YES$ if there is a \emph{separable} quantum state with energy at most $a$, and the answer is $\NO$ if \emph{all separable} quantum states have energy at least $b$.  The \ssh\ problem is defined similarly, but the Hamiltonian is not necessarily local, but rather sparse. We show that the \ssh\ problem is $\QMATC$, while \slh\ is in $\QMA$. This should be compared to the \lh\ problem, and the \sh\ problem which are both $\QMAC$. To the best of our knowledge, \ssh\ is the first non-trivial problem shown to be $\QMATC$. \onote{I added that this is the first $\QMATC$ problem.Is there any reason to believe that it is not correct?).} 
\end{abstract}

\section{Introduction and Results}
\subsection{Introduction}
The class $\QMA$ is the the quantum analogue of the class $\NP$ (or more precisely, $\MA$). The class was first studied by Kitaev~\cite{KSV02}, and has been in focus since: see \cite{aharonov2002qns} for a survey, and \cite{osborne2011hamiltonian} for a more recent physics-motivated review. \fullversion{
First, Watrous showed that $\GROUPNONMEMBERSHIP $ is in $\QMA$~\cite{Wat00} (this problem is not known to be in $\MA$). Then, after a series of works, Kempe, Kitaev and Regev showed that the 2-$\lh$ problem is $\QMA$-complete~\cite{KKR04}. Mariott and Watrous also proved a strong amplification result of $\QMA$~\cite{MW05}. More recently, Aharonov $\etal$ tried to extend the celebrated PCP theorem to the quantum case~\cite{AALV09}. $\QMA$, as the quantum equivalent of $\NP$, is one of the most studied classes in quantum complexity.
}

One of the striking results in proof systems is that sometimes, \emph{limiting} the prover can \emph{increase} the power of the proof system. For example $\IP = \PSPACE$ \cite{lund1992algebraic,shamir1992ip}, while $\MIP = \NEXP$ \cite{babai1991non}. This means that two classical provers can prove more languages to a verifier if it is guaranteed that the provers cannot communicate with each other. However, these classical examples require interaction between the prover and the verifier. The class $\QMA(k)$, introduced by Kobayashi \etal\cite{kobayashi2003quantum}, deals with quantum non-interactive proofs and limits the prover to send $k$ non-entangled proofs, or equivalently $k$-unentangled provers that cannot communicate with each other. The question whether $\QMA(k) = \QMAT$ was answered in the affirmative by Harrow and Montanaro~\cite{harrow2010efficient}. The question whether $\QMAT \subseteq \QMA$ is still open. Note that in the classical case, $\MA(k) = \MA(2) = \MA$. 

To show the power of unentangled quantum proofs, Blier and Tapp~\cite{BT09} first presented a $\QMAT$ protocol for an $\NPC$ problem with two quantum witnesses of size $O(\log(n))$. The drawback of this protocol is that the soundness parameter is somewhat disappointing (1 - $\Omega(1/n^6)$)\onote{I replaced low with bad. The smaller the parameter $s$ is the better. the meaning here is that $s$ is very close to 1. Low soundness error sounds like a good thing to me...}. This was first improved by Beigi~\cite{Bei10} who showed that the soundness can be reduced to $1 - 1/n^{3 + \eps}$ for any $\eps > 0$. Very recently, Le Gall improved this soundness to $1 - \Omega(\frac{1}{n\log(n)})$~\cite{LNN11}. Aaronson \etal showed that there exists a short proof for $\sat$ in $\QMA(\tilde{O}(\sqrt{n}))$~\cite{ABD+08}, where each unentangled witness has logarithmic size, but where the soundness can be exponentially small. In \cite{harrow2010efficient} it was shown that $\sat \in \QMA(2)$, where the size of each proof is $\tilde{O}(\sqrt{n})$. These results tend to show that quantum unentangled proofs are very powerful, since they can solve $\NPC$ problems in a seemingly more efficient way than in $\QMA$. Liu \etal have shown that $\purenrepresentability$, an important problem in quantum chemistry, is in $\QMAT$\cite{liu2007quantum}. This problem is not known to be in $\QMA$.

On the other hand, Brand\~ao~\etal~\cite{BCY11} showed that if the verifier is restricted to performing a Bell measurement\onote{is this correct? I don't recall whether there is only one bell measurement, or multiple bell measurements...}, then, the resulting class $\BELLQMAT$ is equal to $\QMA$. Trying to understand the relationship between $\QMA$ and $\QMAT$ is a fundamental open problem from the point of view of quantum complexity as well as for the understanding of the power of quantum unentangled proofs.

\subsection{Contribution}
In this paper, we study the relationship between $\QMA$ and $\QMA(2)$ from a different perspective. We study the \lh\ problem with unentangled witnesses. The $\klh$ \fullversion{(see Def.~\ref{def:lh})} problem is the quantum analog of $\maxksat$, and is the canonical $\QMAC$ problem.
The first proof that $\klh$ is $\QMAC$ is by Kitaev. Our first result is to extend this construction to separable witnesses in order to find a complete problem for $\QMAT$.
The main ingredient in showing that the $\klh$ problem is $\QMAC$, is Kitaev's Hamiltonian, a Hamiltonian which penalizes states that are not history states. History states are states of the form $\ket{\eta_\psi}\equiv \frac{1}{\sqrt{T+1}}\sum_{t=0}^T \ket{t} \tensor \ket{\psi_t}$, where $\ket{\psi_t}$ is the state at the $t$-th step of the verification process when starting with $\ket{\psi}$ and the  the $m$ ancilla qubits in 0 state, i.e. $\ket{\psi_t} = U_t U_{t-1}\ldots U_0 (\ket{0^m}  \tensor \ket{\psi})$, and $U_i$ is the $i$-th gate used in the $\QMA$ verification circuit, and we set as a convention $U_0 = I$. 

It is natural to try to adapt this argument to a $\QMAT$ verification circuit by constructing a \slh\ problem: the input for the \slh\ problem is the same as the \lh\ problem, i.e. a collection of local Hamiltonians $\{H_1,\ldots,H_m\}$, the answer is $\YES$ if there is a \emph{separable} quantum state with energy at most $a$, and the answer is $\NO$ if \emph{all separable} quantum states have energy at least $b$ for some energies $a < b$. Yet, there is a flaw in this argument: even if $\ket{\psi}= \ket{\chi_A}\tensor\ket{\chi_B}$, the history state $\ket{\eta_\psi}$ might not be separable. \fullversion{This is caused by two reasons:(i) even though $\ket{\psi_0} = \ket{\chi_A}\tensor\ket{\chi_B}$ is a tensor product, for $t>0$, $\ket{\psi_t}$ can be entangled, and (ii) even if $\ket{\psi_t}$ is not entangled, the fact that $\ket{\eta_{\psi}}$ is a superposition over all time steps creates entanglement, as long as both parts of the proof change during the computation.}

In order to resolve the entanglement issue in $\ket{\eta_\psi}$, we use the construction of Harrow and Montanaro~\cite{harrow2010efficient}. 
They show that every $\QMA(k)$ verification circuit can be transformed into a $\QMAT$ circuit with the following structure: The first and second witnesses(which are promised to be non-entangled) have the same length, where each witness contains $r$ registers, where each register size, in the first and second witnesses, is the same. The first $r$ steps of the verification procedure are swap-tests between the $i$-th register of the first and second witnesses, and from that point, the verification circuit acts non-trivially only on the first witness.
In a $\YES$ instance, there exists a non-entangled proof, where $\ket{\chi_A} = \ket{\chi_B} = \ket{\chi_1}\tensor\ket{\chi_2} \tensor \ldots \tensor \ket{\chi_r}$.
Notice that $C-SWAP (\ket{+} \tensor \ket{\phi} \tensor \ket {\phi}) = \ket{+} \tensor \ket{\phi} \tensor \ket {\phi}$, therefore, applying the swap-tests to the above witnesses does not change the state. Since there are no other operations on the second witness, the second witness remains fixed during the entire verification process. If we treat the clock, ancilla qubits and the first witness as the $A$ system, and the second witness as the $B$ system, we get that the history state $\ket{\eta}$ is indeed separable with respect to this division. This is only true if the controlled swap operation is applied on \emph{all} the qubits in the $i$-th register of the first and second witnesses. This will make \onote{new}the propagation terms\fullversion{\footnote{See Eq.~\eqref{eq:prop} for the definition.}} in Kitaev's Hamiltonian non-local. But, on the other hand, a controlled swap operation on arbitrary number of qubits is always sparse: each row has one non-zero entry. This makes each propagation term sparse.

Given a sparse Hamiltonian $H$, the unitary $U = \exp(-iHt)$ can be implemented efficiently, which eventually leads to $\ssh \in \QMAT$. Together with the idea above, it can be shown that:
\begin{theorem}
\label{thm:ssh_qmatc}
$\ssh$ is $\QMATC$.
\end{theorem}    

The only reason why, this construction does not lead to a \slh\ instance, is that the controlled swap gate must be performed in one step; otherwise, $\ket{\eta}$ would become entangled.
At first glance, this might seem as a technicality, but we, surprisingly, show that\onote{Isn't this a little arrogant?}:
\begin{theorem}
\label{thm:slh_in_qma}
\slh\ is $\QMAC$.
\end{theorem}
Since the \slh\ problem is at least as hard as the \lh\ problem, and \lh\ is $\QMAC$, therefore \slh\ is $\QMA$-hard.
To show that $\slh \in \QMA$, we use the \cldm\ problem~\cite{liu_consistency_2006} as a subroutine. Informally, the \cldm\  promise problem asks the following question: given a collection of local density matrices $\rho_i$ over a constant set of qubits $C_i$, is there a quantum state $\rho$ such that for each $i$, the reduced density matrix of $\rho$ over the qubits $C_i$ is equal to $\rho_i$? Liu showed that this problem is $\QMAC$.

To show that \slh\ is $\QMAC$, we do as follows. Assume that there exists a state $\sigma = \sigma_A \tensor \sigma_B$ of total length $2n$, with energy below the threshold $a$. Let $\spa{A},\spa{B}$ the two spaces of qubits considered, each of size $n$. The energy is $\tr{H (\sigma_A \tensor \sigma_B)}$ where $H = \sum_i H_i$. Let $C_i$ the subset of qubits each $H_i$ act on. \onote{I changed $\sigma_{C_{i}}$ to $\sigma^{C_{i}}$. This is the way used in Nielsen and Chuang.}We have $\tr{H (\sigma_A \tensor \sigma_B)} = \sum_{i=1}^m \tr{H_i \sigma^{C_i}}$, where $\sigma^{C_i}$ corresponds to the reduced state of $\sigma$ on the qubits of $C_i$. Again, we can decompose $\sigma^{C_i}$ into the A part and the B part. We can write $\sigma^{C_i} = \sigma^{A_i} \otimes \sigma^{B_i}$. This is because the state $\sigma$ is a product state between $\spa{A}$ and $\spa{B}$, hence, the state $\sigma^{C_i}$ is also a product state between $\spa{A}$ and $\spa{B}$.

The proof will consist of a classical part: the classical description of the reduced density matrices $\sigma^{A_i},\sigma^{B_i}$. This information is sufficient to calculate the energy classically, using $\tr{H (\sigma_A \tensor \sigma_B)} = \sum_{i=1}^m \tr{H_i (\sigma^{A_i} \otimes \sigma^{B_i})}$.
The proof also consists of a quantum part: the prover tries to convince the verifier that there exists a quantum mixed state $\rho_A$ and similarly for $\rho_B$ that are consistent with the reduced density matrices $\sigma^{A_i}$ and $\sigma^{B_i}$. Since \cldm\ is known to be in $\QMA$, the prover can convince the verifier if there exists such a state, but cannot fool the verifier if there is no such state. 

\paragraph{Discussion}
\abstractversion{In the setting of $\QMA$, both the \lh\ and the \sh\ are natural $\QMAC$ problems. When we consider separable witnesses, \slh\ and \ssh\ seem to be natural $\QMATC$ problem. Theorem~\ref{thm:ssh_qmatc} proves that \ssh\ is indeed $\QMATC$, in sharp contrast to the \slh\ problem, which is shown to be in $\QMA$, by Theorem~\ref{thm:slh_in_qma}.
}

\fullversion{
In the case, where we do not consider separable witnesses, the two problems \lh\ and \sh\ are natural $\QMAC$ problems. Thus, in this setting, considering sparse Hamiltonians instead of local Hamiltonians does not increase the power of the verifier.

When we consider separable witnesses, things are different. \slh\ and \ssh\ seem to be natural $\QMATC$ problems. With Theorem~\ref{thm:ssh_qmatc}, we show that \ssh\ is indeed $\QMATC$ by adapting Kitaev's completeness and using the new construction from Harrow and Montanaro~\cite{harrow2010efficient}. However, we were not able to remove this sparseness condition to show that \slh\ is also $\QMATC$.

On the other hand, we show that \slh\ is $\QMAC$. We find this surprising because \slh\ was a natural candidate for a $\QMATC$ problem. This also means that when considering separable witnesses, the sparse condition for Hamiltonians is crucial or conversely that separable witnesses do not help a verifier when his accepting procedure is a sum of local Hamiltonians. 
This is in sharp contrast with the general case where separable witnesses seem to help the verifier significantly\onote{notice that the ``general case'' was used the opposite way before. I think that this is fine.}.
While we do not have a clear separation between $\QMA$ and $\QMAT$, we know that \onote{old: $\QMA \in \PP \in \PSPACE$} $\QMA \subseteq \PP \subseteq \PSPACE$ (first unpublished proof by Kitaev and Watrous then simplified in ~\cite{MW05}) while we only know that \onote{old:$\QMA \in \QMA(2) \in \NEXP$}$\QMA \subseteq \QMA(2) \subseteq \NEXP$~\cite{kobayashi2003quantum}.

Our results characterize rather tightly the difference between $\QMA$ and $\QMAT$. We hope that this will lead to a better understanding of the relationship between these classes.}

\fullversion{
\paragraph{Structure of the paper:}  Section~\ref{sec:definitions} contains the preliminaries and definitions. In Section~\ref{sec:sparse}, we show that \ssh\ is $\QMATC$ (Theorem~\ref{thm:ssh_qmatc}). In Section~\ref{sec:local}, we show that \slh\ is $\QMAC$ (Theorem~\ref{thm:slh_in_qma}).

\section{Preliminaries and Definitions}
\label{sec:definitions}
\begin{definition}
\label{def:qmat}
 A promise problem $L=\{L_{yes},L_{no}\}$ is in $\QMA_{s,c}(k)$ if there exists a uniformly generated polynomial time quantum algorithm $\mathcal{A}$ and computable polynomially bounded functions $f_1,\ldots,f_k$ such that for all input $x \in \{0,1\}^n$:
 \begin{enumerate}
	\item {\bf Completeness:} if $x \in L_{yes}$ there exist $k$ witnesses $\ket{\psi_1}, \ldots,\ket{\psi_k}$, where each witness $\ket{\psi_i}$ consists of $f_i(n)$ qubits such that $\mathcal{A}$ accepts $\ket{x} \tensor \ket{\psi_1} \tensor \ldots \tensor \ket{\psi_k}$ with probability at least $c$.
	\item {\bf Soundness:} if $x \in L_{no}$ then for all $k$ witnesses $\ket{\psi_1},  \ldots, \ket{\psi_k}$, where each witness $\ket{\psi_i}$ consists of $f_i(n)$ qubits the probability that $\mathcal{A}$ accepts $\ket{x} \tensor \ket{\psi_1} \tensor \ldots \tensor \ket{\psi_k}$ is at most $s$.
\end{enumerate}
\end{definition}
We define $\QMA_{\frac{1}{3},\frac{2}{3}}(k) = \QMA(k)$, and $\QMA=\QMA(1)$.

\begin{theorem}[\cite{harrow2010efficient}]
\label{thm:qma_k_equals_qma_2}
If $c-s \geq 1/poly(n)$, $k=poly(n)$, and $p(n)$ is an arbitrary polynomial, then $\QMA_{s,c}(k) = \QMA_{2^{-p(n)},1-2^{-p(n)}}(2)$.

Furthermore, it can be assumed w.l.o.g. that the $\QMAT$ protocol has the following structure: The two witnesses have exactly the same size, where both of them consist of $r$ registers of sizes $s_1,\ldots,s_r$. The verification process consists of applying the product test (see Def.~\ref{def:product_test}). If the product test fails, Arthur rejects. Otherwise, Arthur runs a polynomial quantum algorithm $\mathcal{A}$ on the first proof, and outputs the result. In a $\YES$ instance, the two Merlins can send identical states, which are tensor product between the $r$ registers: $\ket{\psi_1}=\ket{\psi_2} = \ket{\chi_1}\tensor \ldots \tensor \ket{\chi_r}$.
\end{theorem}
\begin{definition}[Product Test\cite{harrow2010efficient}]
\label{def:product_test}
The input consists of two states, where each state has $r$ registers of size $s_1,\ldots,s_r$. 
Preform the swap test on each of the $r$ pairs. Accept if all of the swap-tests pass, otherwise reject.
\end{definition}

\begin{definition}[\klh\ problem]
\label{def:lh} 
Input: a set of hermitian matrices $H_1,\ldots,H_m$, where each matrix operates on a set of at most $k$ out of the $n$ qubits, and $I \succeq H_i \succeq 0$ (i.e. both $H_i$ and $I-H_i$ are positive semi definite), and two real number $a$ and $b$ such that $b-a > poly(1/n)$. We define the Hamiltonian, with a slight abuse of notation\footnote{Each matrix $H_i$ operates on some set of qubits, and the summation is over their extension to the entire Hilbert space of the $n$ qubits.}, $H=\sum_{i=1}^m H_i$.
Output: Output $\YES$ if there exists a state $\ket{\psi}$ such that $\bra{\psi}H\ket{\psi} \leq a$, and $\NO$ if for every state $\ket{\psi}$, $\bra{\psi}H\ket{\psi} \geq b$. 
\end{definition}
\begin{definition}[Simulatable Hamiltonian\cite{AT03}]
\label{def:simulatable_hamiltonian}
We say a Hamiltonian $H$ on $n$ qubits is simulatable if for every $t>0$ and every accuracy $\alpha > 0$, the unitary transformation $U = \exp(-i H t)$ can be approximated to within $\alpha$ accuracy by a quantum circuit of size $poly(n,t,\frac{1}{\alpha})$.  
\end{definition}

\begin{definition}[\sklh\ problem]
\label{def:slh}
The input is the same as the input for the \klh\ problem together with a partition of the qubits to disjoint sets $A$ and $B$. The answer is $\YES$ if $\exists \ket{\psi} = \ket{\chi_{A}} \tensor \ket{\chi_{B}} \ s.t. \  \bra{\psi} H \ket{\psi} \leq a $ and the answer is $\NO$ if  $\bra{\psi} H \ket{\psi} \geq b$ for all tensor product states  $\ket{\psi} = \ket{\chi_{A}} \tensor \ket{\chi_{B}}$.
\end{definition}

\paragraph{Remark:} The above definition can be formulated using mixed states in the two following ways, with mixed product states and mixed separable states. It can be verified that indeed these definitions are equivalent.

\begin{definition}[\sklh\ problem - alternative definition 1]\label{def:AlternativeDefinition1}
The input is the same as in Def.~\ref{def:slh}. The answer is $\YES$ if there exists a product mixed state $\rho = \rho_{A} \tensor \rho_{B}$ s.t. $\tr{ H \rho} \leq a $ and the answer is $\NO$ if  $\tr{ H \rho}  \geq b$ for all product mixed states $\rho = \rho_{A} \tensor \rho_{B}$.
\end{definition}

\begin{definition}[\sklh\ problem - alternative definition 2]
The input is the same as in Def.~\ref{def:slh}. The answer is $\YES$ if there exists a separable mixed state $\rho = \sum_{i} p_{i} (\rho_{i}^{A}\tensor \rho_{i}^{B}) \ s.t. \  \tr{ H \rho} \leq a $ and the answer is $\NO$ if  $\tr{ H \rho}  \geq b$ for all separable mixed states $\rho = \sum_{i} p_{i} (\rho_{i}^{A}\tensor \rho_{i}^{B}) $.
\end{definition}

We now define the \ssh\ problem.

\begin{definition}[\ssh]
\label{def:ssh}
An operator $A$ over $n$ qubits is row-sparse if each row in $A$ has at most $poly(n)$ non-zero entries, and there exists an efficient classical algorithm that, given $i$, outputs a list $(j, A_{i,j})$ running over all non zero elements of $A_{i,j}$. The \ssh\ problem is the same as \sklh\ except each term in the input Hamiltonian is row-sparse instead of $k$-local.
\end{definition}

Finally, we define the \cldm\ problem which we will use to show that the \sklh\ problem is $\QMAC$.

\begin{definition}[\cldm\ \cite {liu_consistency_2006}]
\label{def:cldm}
We are given a collection of local density matrices $\rho_1,\ldots,\rho_m$, where each $\rho_i$ is a density matrix over qubits $C_i \subset \{1,\ldots,n\}$, and $|C_i| \leq k$ for some constant $k$. Each matrix entry is specified by $poly(n)$ bits of precision. In addition, we are given a real number $\beta \geq 1/poly(n)$ specified with $poly(n)$ bits of precision. The problem is to distinguish between the following two cases:
\begin{enumerate}
	\item There exists an $n$ qubits mixed state $\sigma$ such that for all $i$ such that $ \Tr_{\{1,\ldots,n\} \setminus C_i} (\sigma) = \rho_i$. In this case, output $\YES$.
	\item For all $n$ qubits mixed states $\sigma$, there exists some $i$ such that 
$|| \Tr_{\{1,\ldots,n\} \setminus C_i} (\sigma) - \rho_i ||_1 \geq \beta$. In this case output $\NO$.
\end{enumerate} 
\end{definition}

\section{\texorpdfstring{Proof that \ssh\ is $\QMATC$}{Proof that Separable Simulatable Hamiltonian is QMA(2)-complete}}\label{sec:sparse}
\subsection{\texorpdfstring{$\ssh \in \QMAT$}{Separable Simulatable Hamiltonian is in QMA(2)}}
The construction has the same structure as the proof that $\lh \in \QMA$ in \cite{KSV02}, and uses phase estimation as a subroutine to achieve that
\cite{KSV02,NC00}. We consider row sparse Hamiltonians $\{H_j\}_{1 \leq j \leq m}$ and $H = \sum_{j=1}^{m} H_j$. For each $j$, we construct a quantum algorithm $Q_j$ such that
\[ | \pr(Q_j \ accepts \ \ket{\psi}) - (1 - \bra{\psi}H_j\ket{\psi})| \le \eps,\]
where we choose $\eps = \frac{b-a}{3}$, and the running time of $Q_j$ is polynomial in $n$.
Let $Q$ be the algorithm where we pick $1 \leq j \leq m$ at random, and run $Q_j$. 
\[ | \pr(Q \ accepts \ \ket{\psi}) - (1 - \frac{1}{m}\bra{\psi}H\ket{\psi})| \le \eps. \]
Therefore, in a $\YES$ instance there exists a state $\ket{\psi} = \ket{\chi_A} \tensor \ket{\chi_B}$ which is accepted with probability at least $c=1-\frac{a}{m} - \eps$, whereas in a $\NO$ instance, the probability of acceptance for every tensor product state is at most $s=1-\frac{b}{m} + \eps$. Therefore the problem is in $\QMA_{c,s}(2)$, which is equal to $\QMAT$ by Thm.~\ref{thm:qma_k_equals_qma_2}.

All that is left to show how to implement $Q_j$. Aharonov and Ta-Shma have shown:
 
 \begin{lemma}[The sparse Hamiltonian lemma\cite{AT03}]
 \label{le:spare_hamiltonian}
 If $H$ is row-sparse, and $||H|| \leq poly(n)$ then $H$ is simulatable. 
 \end{lemma}
 
 \begin{theorem}[Phase Estimation\cite{NC00}]
 \label{thm:phase_estimation}
Let $V$ be a unitary which can be implemented by a quantum circuit with $d$ gates, which has eigenstates $\{\ket{u_j}\}_{1\leq j \leq N}$, and eigenvalues $\{ e^{i \phi_j}\}_{1\leq j \leq N}$. Given a state $\ket{\phi} = \sum_{i=1}^N \sqrt{p_i} \ket{u_i}$, an error parameter $\epsilon$ and a precision parameter $\delta$, the phase estimation procedure outputs with probability at least $p_i(1 - \eps)$ a number which is $\delta$ close to $\phi_i$.
 
 Let $t = \log(\delta) + \ceil{\log(2+ \frac{1}{2 \epsilon})}$. The phase estimation procedure can be implemented by a quantum circuit with $O(t^2 + d^{2^t})$ gates. 
 \end{theorem}




We can now show how to implement 
$Q_j$:
\begin{enumerate}
	\item Start with a state $\ket{\psi} = \sum_{i=1}^N \sqrt{p_i} \ket{u_i}$, where $\ket{u_i}$ is an eigenstate of $H_j$ with eigenvalue $\phi_i$. 
	\item Apply phase estimation with the unitary $U=\exp(iH_j)$ with probability for an error and precision $\frac{b-a}{6}$. Denote by $\tilde{\phi}$ the output of the phase estimation. 
\item Reject with probability $\tilde{\phi}$. 

\end{enumerate}
Using Thm.~\ref{thm:phase_estimation}, we can get both a lower and an upper bound on the acceptance probability:
\[(1- \frac{b - a}{6}) \sum_{i=1}^N p_i (\phi_i - \frac{b - a}{6}) \leq Pr(Q_j\ rejects \ \ket{\psi}) \leq \frac{b-a}{6} + (1- \frac{b - a}{6})\sum_{i=1}^N p_i (\phi_i + \frac{b - a}{6}). \]
Since $\sum_{i=1}^N p_i \phi_i = \bra{\psi} H_j \ket{\psi}$, and $I \succeq H_i \succeq 0$ we get:
\[  \bra{\psi} H_j \ket{\psi} -  \frac{b-a}{3}  \leq Pr(Q_j\ rejects \ \ket{\psi}) \leq  \bra{\psi} H_j \ket{\psi} + \frac{b-a}{3},\]
which was the requirement for $Q_j$.

Unfortunately, By Lemma~\ref{le:spare_hamiltonian}, we can only approximate $U=\exp(iH_j)$ which is needed in step 2. A polynomial approximation, which can be achieved in polynomial time, is good enough for our needs, for similar reasons as the analysis done in \cite{AT03} and \cite[Section 4.1]{WZ06}.
\onote{Hopefully, add explicit calculation in next version.}


\subsection{\texorpdfstring{\ssh\ is $\QMATH$}{Separable Simulatable Hamiltonian is QMA(2)-hard}}
Consider a promise problem $L = \{L_{yes},L_{no}\}$ which is in $\QMA_{s,c}(2)$ with $c = 1 - \frac{C}{512(T+1)^4}$ and $s = \frac{1}{T+1}$. $C$ is a universal constant that will be specified later. For such $s$ and $c$, we have $\QMA_{s,c}(2) = \QMA(2)$ by Theorem~\ref{thm:qma_k_equals_qma_2}\onote{old: This problem is hard for $\QMA(2)$ by theorem 2. Or: we don't need to assume that this is QMA-hard problem.}. Our goal is to reduce this problem to the \ssh\ problem.

Pick an instance $x$ of $L$ and let $A$ the associated verifying procedure. We will omit the dependence in $x$ and write the verifying procedure as a unitary $U$ taking as input the two quantum witnesses. We can assume w.l.o.g. that this verification procedure has the structure described in Thm.~\ref{thm:qma_k_equals_qma_2}. We decompose the verifying procedure into $T$ unitaries $U = U_1,\dots,U_T$ each acting on a 2 qubits. This means that after $t$ steps of the verifying procedure, the unitary applied is $U_t U_{t-1} \cdots U_0$\onote{the index was k. The index k sometimes refers to the locality, sometimes to the number of provers, so I don't want it to have a third meaning. I changed it to t.}, where we add the convention that $U_0 = I$.

We apply Kitaev's construction (See \cite[Sec. 14.4.1]{KSV02} for the detailed definition) of the circuit, and get a Hamiltonian of the form\footnote{Although the unary clock can be implemented, we use the construction where the clock is implemented using $O(\log(n))$ qubits for simplicity.}:
\[H = H_{in} + H_{prop} + H_{out}. \]

It should be stressed that the swap-test is implemented in a non-local manner, and therefore, $H_{prop}$ is not local. Nevertheless, each term in $H_{prop}$ is sparse. Reminder:
\begin{equation}
H_{prop} = \sum_{t=1}^T H_t ,
\label{eq:prop}
\end{equation}
\begin{equation}
 H_t = -  \frac{1}{2}  \ket{t}\bra{t-1}\tensor U_t  -   \frac{1}{2} \ket{t-1}\bra{t} \tensor U_t^\dagger  + \frac{1}{2}(\ketbra{t}{t} + \ketbra{t-1}{t-1}) \tensor I . 
\label{eq:prop_term}
\end{equation}

 Indeed, it can be verified that $H_t$ has at most $2$ non-zero entries in each row, in the case that $U_t = C-SWAP$, regardless of the size of the swapped registers.  To prove our reduction, we show the following:
\begin{itemize}
\item If $x \in L_{yes}$ then there exists $\ket{\psi} = \ket{\psi_1} \otimes \ket{\psi_2}$, such that $\bra{\psi} H \ket{\psi} \le \frac{C}{512(T+1)^5}$.
\item If $x \in L_{no}$ then for all $\ket{\psi} = \ket{\psi_1} \otimes \ket{\psi_2}$, $\bra{\psi} H \ket{\psi} \ge \frac{C}{256(T+1)^5}$.
\end{itemize}
\paragraph{Completeness:} In a $\YES$ instance, the two Merlins can send identical states which are accepted with probability at least $c$ (where $c$ is the completeness parameter), which have the form $\ket{\psi_1}=\ket{\psi_2}=\ket{\chi_1} \tensor \ldots \tensor \ket{\chi_r}$. Since $C-SWAP (\ket{+} \tensor \ket{\chi_i}\tensor \ket{\chi_i}) = \ket{+}  \tensor \ket{\chi_i}\tensor \ket{\chi_i}$, the first $r$ steps of the verification protocol (see Thm.~\ref{thm:qma_k_equals_qma_2})) have no effect. Therefore,
\begin{align*}
\ket{\eta} &=\frac{1}{\sqrt{T+1}}\sum_{t=0}^T \ket{t}^C \tensor U_t U_{t-1}\ldots U_0  \left(\ket{0^m}^A \tensor \ket{\psi_1}^{P_1} \tensor \ket{\psi_2}^{P_2}\right)
\end{align*}
where $C$ is the clock subsystem, $A$ is the ancilla subsystem, and $P_{1}$ and $P_{2}$ are the first and second proof subsystems. 

Using Theorem~\ref{thm:qma_k_equals_qma_2}, we know that in a $\YES$ instance, the two Merlins can send identical states which are accepted with probability at least $c$. These states are of the form $\ket{\psi_1}=\ket{\psi_2}=\ket{\chi_1} \tensor \ldots \tensor \ket{\chi_r}$. The verifier then does the following: he first performs a swap test on each pair $\ket{\chi_i}\tensor \ket{\chi_i}$ (characterized by the first $r$ unitaries $U_1,dots,U_r$). This does not change the state at all. He then applies the verifying procedure only on the first proof, and ancilla. Therefore,
\begin{align*}
\ket{\eta} &=\frac{1}{\sqrt{T+1}}\sum_{t=0}^T \ket{t}^C \tensor U_t U_{t-1}\ldots U_0  \left(\ket{0^m}^A \tensor \ket{\psi_1}^{P_1} \tensor \ket{\psi_2}^{P_2}\right) \\
&= \left(\frac{1}{\sqrt{T+1}}\sum_{t=r+1}^T \ket{t}^C \tensor U_t U_{t-1}\ldots U_{r+1}  \left(\ket{0^m}^A \tensor \ket{\psi_1}^{P_1}\right) \right) \tensor \ket{\psi_2}^{P_2}.
\end{align*}
This shows that $\ket{\eta}$ is a tensor product state with respect to the spaces ($C \otimes A \otimes P_1$) on one end and $P_2$ on the other. Kitaev's proof (see \cite[Sec. 14.4.3]{KSV02}) shows that 
\begin{equation}
\bra{\eta}H\ket{\eta} \leq \frac{1-c}{T+1},
\end{equation}
and by substituting $c$, we get,
\begin{equation}
\label{eq:completeness}
\bra{\eta}H\ket{\eta} \leq \frac{C}{512(T+1)^5}.
\end{equation}

\paragraph{Soundness:} 
We first outline the three steps of the proof qualitatively. We assume that there exists a low-energy state $\ket{\omega} = \ket{\omega_1} \tensor \ket{\omega_2}$, and we show that:

(i) If $\ket{\omega}$ has low energy, then $\ket{\omega}$ is close to a history state $\ket{\eta}$, i.e. a state of the form $\frac{1}{\sqrt{T+1}} \sum_{t=0}^T \ket{t} \tensor U_t U_{t-1}\ldots U_0  \ket{0^m} \tensor \ket{\psi}$ for some $\ket{\psi}$. We usually write such a history state $\ket{\eta_{\psi}}$ to mark the dependence in $\ket{\psi}$.

(ii) If a history state $\ket{\eta_{\psi}}$ is close to a product state (and by (i), it is), then, the associated state $\ket{\psi}$ is close to a tensor product state.

(iii) If $\ket{\psi}$ is close to a tensor product state, the originating history state $\ket{\eta_\psi}$ must have high energy which will contradict (i).

\begin{lemma}[Step one]
\label{le:step_one} If $\bra{\omega} H \ket{\omega} \leq \alpha \frac{C}{(T+1)^3}$, for some universal constant $C$, then, there exists a history state $\ket{\eta}$ s.t. $|\braket{\omega}{\eta}|^2 \geq 1 - \alpha$.
\end{lemma}
\begin{proof}
Let $\mathcal{V}_{hist}$ the subspace spanned by all history states. We can verify that $\mathcal{V}_{hist}$ is the kernel of $H_{init} + H_{prop}$ \anote{ref ?}. We use the following Claim, which is proved in Appendix~\ref{app:spectral_gap}.
\begin{claim}
\label{cl:gap}
 $\Delta(H_{init} + H_{prop}) \geq \frac{C}{(T+1)^3}$, for some universal constant $C$, where $\Delta(A)$ is the smallest non-zero eigenvalue of $A$. 
\end{claim}
We can write $\ket{\omega} = \sqrt{1-p} \ket{\eta} + \sqrt{p} \ket{\eta^\perp}$, for $\ket{\eta} \in \mathcal{V}_{hist}$, and $\ket{\eta^\perp} \in \mathcal{V}_{hist}^\perp$ . By assumption
\begin{align*}
 \alpha \frac{C}{(T+1)^3} \geq \bra{\omega} H \ket{\omega}\geq  \bra{\omega} H_{init}+H_{prop} \ket{\omega} = p \bra{\eta^\perp} H_{init}+H_{prop} \ket{\eta^\perp} \geq p\frac{C}{(T+1)^3}, 
\end{align*}
where the first inequality follows from the assumption of the lemma, the second uses the fact that $H_{out} \succeq 0$ and the last inequality uses Claim~\ref{cl:gap}. 
To conclude, $p \leq \alpha$ which implies that $|\braket{\omega}{\eta}|^2 = 1 - p \geq 1 - \alpha$, as needed.
\end{proof}

\begin{lemma}[Step two]
\label{le:step_two}
Let $\ket{\eta_{\psi}}$ a history state and $\ket{\psi}$ such that 
\begin{align*}
\ket{\eta_{\psi}} & =  \frac{1}{\sqrt{T+1}} \sum_{t=0}^T \ket{t}^C \tensor U_t U_{t-1}\ldots U_0  (\ket{0^m}^A \tensor \ket{\psi}^{P_1,P_2}) \\
& = \frac{1}{\sqrt{T+1}} \ket{0}^{C}\ket{0^m}^{A}\ket{\psi}^{P_1,P_2} +
\frac{1}{\sqrt{T+1}} \sum_{t=1}^T \ket{t}^C \tensor U_t U_{t-1}\ldots U_1  (\ket{0^m}^A\ket{\psi}^{P_1,P_2})  
\end{align*}
where we consider the following subsystems:  $C$ is the clock subsystem, $A$ is the ancilla subsystem, and $P_{1}$ and $P_{2}$ are the first and second proof subsystems. 
If there exist two states $\ket{\psi_1} \in \spa{C} \otimes \spa{A} \otimes \spa{P}_1$ and $\ket{\psi_2} \in \spa{P}_2$ such that 
$
|\bra{\eta_\psi} (\ket{\psi_1} \tensor \ket{\psi_2})|^2 \geq 1 - \eps
$
then there exists a state $\ket{L} \in \spa{P}_1$ such that $|\bra{\psi} (\ket{L}\tensor \ket{\psi_2})|^2 \geq 1 - \eps(T+1)$
\end{lemma}
\begin{proof}
We write
\begin{align*}
\ket{\psi_1} = \sqrt{\alpha} \ket{0}^C \ket{0^m}^A \ket{L}^{P_1} + 
\sqrt{1 - \alpha}\sum_{i,j : (i,j) \neq (0,0^m)} \beta_{i,j}\ket{i}^C \ket{j}^A \ket{\psi_{i,j}}^{P_1}
\end{align*}
with $\sum_{i,j : (i,j) \neq (0,0^m)} |\beta{i,j}|^2 = 1$. From this, we immediately have
\begin{align*}
| \bra{\eta_\psi} (\ket{\psi_1} \otimes \ket{\psi_2}) | & \le 
\sqrt{\frac{\alpha}{T+1}}\cdot |\bra{\psi} (\ket{L}\tensor \ket{\psi_2})| + \sqrt{\frac{T(1 - \alpha)}{T+1}} \cdot 1 \\
& = \sqrt{\alpha} \sqrt{\frac{1}{T+1}}\cdot |\bra{\psi} (\ket{L}\tensor \ket{\psi_2})| + \sqrt{1 - \alpha} \sqrt{\frac{T}{T+1}} \\
 & \le \sqrt{\frac{|\bra{\psi} (\ket{L}\tensor \ket{\psi_2})|^2}{T+1} + \frac{T}{T+1}},
\end{align*}
where used Cauchy Schwarz in both inequalities. Therefore,
\begin{align*}
|\bra{\psi} (\ket{L}\tensor \ket{\psi_2})|^2 & \ge
(T+1) \cdot \left(| \bra{\eta_\psi} (\ket{\psi_1} \otimes \ket{\psi_2}) |^2 - \frac{T}{T+1}\right) .
\end{align*}
By using $|\bra{\eta_\psi}  (\ket{\psi_1} \tensor \ket{\psi_2})|^2 \geq 1 -  \eps$, we can further bound
\begin{align*}
|\bra{\psi} (\ket{L}\tensor \ket{\psi_2})|^2 & \ge  1 - \eps(T+1).
\end{align*}
\end{proof}

\begin{lemma}[Step three]
\label{le:step_three}
Consider a history state $\ket{\eta_{\psi}}$ with an associated state $\ket{\psi}$. In a $\NO$ instance with soundness parameter $s$, if $|\bra{\psi} (\ket{\psi_1} \tensor \ket{\psi_2})|^2 \geq 1 - \eps$, then $\bra{\eta_\psi}H\ket{\eta_\psi} \geq \frac{1}{T+1}(1-s - 2 \sqrt{\eps})$.
\end{lemma}
To prove this Lemma, we will need the following Claim.
\begin{claim}
\label{cl:close_rejecting_probability}
Let $\Pi$ be a projector, and $\ket{v_1}, \ket{v_2}$ be arbitrary, and let $q_i = \bra{v_i} \Pi \ket{v_i}$. If $|\braket{v_1}{v_2}|^2 \geq 1 - \delta$, then, $|q_1 - q_2|  \leq \sqrt{\delta}$.
\end{claim}

\begin{proof}
To prove this claim, we use the trace distance between $\ketbra{v_1}{v_1}$ and $\ketbra{v_2}{v_2}$. The trace distance is denoted $D(\ketbra{v_1}{v_1},\ketbra{v_2}{v_2})$ and is equal to $\frac{1}{2}\|\,\ketbra{v_1}{v_1} - \ketbra{v_2}{v_2} \,\|_1$.\onote{Andre: I changed your notation to D instead of $\Delta$ because I'm using $\Delta$ already in the appendix for the spectral gap of an operator.} We know that from the characterization of the trace distance that can be found in~\cite{NC00}, we have 
\begin{align*}
D(\ketbra{v_1}{v_1},\ketbra{v_2}{v_2}) \ge |q_1 - q_2|. \end{align*}
Moreover, we know by a Fuchs- van de Graaf inequality that $D(\ketbra{v_1}{v_1},\ketbra{v_2}{v_2}) \le \sqrt{1 - |\braket{v_1}{v_2}|^2}$ (\cite{FG99}). By putting everything together, we have
\begin{align*}
|q_1 - q_2| \le D(\ketbra{v_1}{v_1},\ketbra{v_2}{v_2})  \le  \sqrt{1 - |\braket{v_1}{v_2}|^2} \le \sqrt{\delta}. \end{align*}
\end{proof}

We can now prove the Lemma.
\begin{proof}
For every state $\ket{\psi}$,
\[ \bra{\eta_\psi} H \ket{\eta_\psi} = \frac{1}{T+1} \bra{0^m} \tensor \bra{\psi} U_1^\dagger \ldots U_t^\dagger \Pi_{reject}U_t \ldots U_1 \ket{0^m} \tensor \ket{\psi}\]

By using Claim \ref{cl:close_rejecting_probability}, we get 
\begin{align*}
\bra{\eta_\psi} H \ket{\eta_\psi} &\geq \frac{1}{T+1} ( \bra{0^m} \tensor \bra{\psi_1} \tensor \bra{\psi_2} U_1^\dagger \ldots U_t^\dagger \Pi_{reject}U_t \ldots U_1 \ket{0^m} \tensor \ket{\psi_1} \tensor \ket{\psi_2} - 2\sqrt{\eps}) \\
 &= \frac{1}{T+1} (\pr(\mathcal{A}\ rejects \ket{\psi_1}\tensor\ket{\psi_2}) - 2\sqrt{\eps}\\
 &\geq \frac{1}{T+1} (1 - s - 2 \sqrt{\eps}),
\end{align*}

where in the last inequality, we used the fact that this is a $\NO$ instance, and therefore all tensor product states are rejected with probability at least $1-s$.
\end{proof}
We can now combine the three steps, and prove the soundness property. Assume, by contradiction, that there exists a state $\ket{\omega}= \ket{\omega_{1}} \tensor \ket{\omega_{2}}$, with energy below the promise, i.e. $\bra{\omega} H \ket{\omega} \leq \frac{C}{256(T+1)^{5}}$.

By Lemma \ref{le:step_one}, there exists a state $\ket{\eta_{\psi}}$ such that 
\begin{equation}
\label{eq:overlap_omega_psi}
|\braket{\eta_{\psi}}{\omega}|^{2} \geq 1 -\frac{1}{256(T+1)^{2}}.
\end{equation}
Using Lemma \ref{le:step_two}, there exists $\ket{\phi} = \ket{\phi_{1}} \tensor \ket{\phi_{1}}$ such that $|\braket{\psi}{\phi}|^{2} \geq 1 -\frac{1}{256(T+1)}$. By Lemma \ref{le:step_three},
\[ \bra{\eta_{\psi}} H \ket{\eta_{\psi}} \geq \frac{1}{T+1}\left(1 - \frac{1}{T+1} - 2 \sqrt{\frac{1}{256(T+1)}}\right). \]

Our goal is to lower bound $\bra{\omega} H \ket{\omega}$. We have, $\ket{\omega} = \sqrt{1-p} \ket{\eta_{\psi}} + \sqrt{p} \ket{\eta_{\psi}^{\perp}}$, for some $\ket{\eta_{\psi}^{\perp}}$ that satisfies $\braket{\eta_{\psi}}{\eta_{\psi}^{\perp}} = 0$, where $ 0 \leq p \leq \frac{1}{256(T+1)^{2}}$ by Eq.~\eqref{eq:overlap_omega_psi}, therefore
\begin{equation}
\bra{\omega}H \ket{\omega}= (1-p) \bra{\eta_{\omega}}H\ket{\eta_{\omega}} + p \bra{\eta_{\omega}^{\perp}}H\ket{\eta_{\omega}^{\perp}} + 2\sqrt{(1-p)p}\, Re(\bra{\eta_{\omega}^{\perp}}H\ket{\eta_{\omega}})
\end{equation}
We define $\delta \equiv\frac{1}{T+1}\left(1 - \frac{1}{T+1} - 2 \sqrt{\frac{1}{256(T+1)}}\right)$. Since, $H \succeq 0$, clearly $\bra{\eta_{\omega}^{\perp}}H\ket{\eta_{\omega}^{\perp}} \geq 0$. Also, $H=H_{in}+H_{prop}+H_{out}$, and $(H_{in}+H_{prop})\ket{\eta_{\psi}} = 0$, and $I \succeq H_{out} \succeq 0$ which implies $Re(\bra{\eta_{\omega}^{\perp}}H\ket{\eta_{\omega}}) = Re(\bra{\eta_{\omega^{\perp}}}H_{out}\ket{\eta_{\omega}}) \geq -1 $. Together, this gives

\[ \bra{\omega}H \ket{\omega} \geq (1-p)\delta - 2 \sqrt{(1-p)p}. \]
We can lower bound the first term by using $\delta \geq \frac{1}{2(T+1)}$ and $p \leq \frac{1}{2}$, and the second term by using $\sqrt{(1-p)p} \leq \sqrt{p}$, hence
\[ \bra{\omega}H \ket{\omega} \geq \frac{1}{4(T+1)} - 2 \sqrt{\frac{1}{256(T+1)^{2}}}  =\frac{1}{8(T+1)}. \]
This contradicts our assumption that $\bra{\omega} H \ket{\omega} \leq \frac{C}{256(T+1)^{5}}$, and proves the soundness property. 

To conclude, when $c=1-\frac{C}{512(T+1)^{4}}$ and $s=\frac{1}{T+1}$, we showed that in a $\YES$ case, there exists a tensor product state with energy at most $\frac{C}{512(T+1)^{5}}$, and in a $\NO$ case, all tensor product states have energy at least $\frac{C}{256(T+1)^{5}}$, which completes the proof of Thm.~\ref{thm:ssh_qmatc}.

\COMMENT{
\onote{compltete the proof by combining the three lemmas with the right parameters.}
\begin{proofof}{Claim~\ref{cl:close_rejecting_probability}}
\begin{fact}[\cite{NC00}]
\label{fa:trace_distance}
\[D(\rho,\sigma) \equiv \frac{1}{2}\tr{|\rho-\sigma|)} = max_\Pi \tr{\Pi (\rho-\sigma)}.\]
\end{fact}
\begin{fact}[\onote{cite N. and C. sec. 9.2.3}]
\label{fa:trace_and_fidelity}
For two pure states, $\ket{a},\ket{b}$,
\[D(\ketbra{a}{a},\ketbra{b}{b}) = \sqrt{1-|\braket{a}{b}|^2} \]
\end{fact}
By definition, $p_i = \tr{\Pi \ketbra{\psi_i}{\psi_i}}$. 
\[p_1-p_2 = \frac{1}{2}\tr{\Pi (\ketbra{\psi_1}{\psi_1}-\ketbra{\psi_2}{\psi_2})} \leq 2 D(\ketbra{\psi_1}{\psi_1},\ketbra{\psi_2}{\psi_2})\]
\[ = 2 \sqrt{1-|\braket{\psi_1}{\psi_2}|^2} \leq 2 \sqrt{\epsilon}, \]
where in the first inequality we used Fact~\ref{fa:trace_distance} and in the second equality we used Fact~\ref{fa:trace_and_fidelity}.
\end{proofof}
}
\section{\texorpdfstring{Proof that \sklh\ is $\QMAC$}{Proof that Separable k-local Hamiltonian is QMA-Complete}}\label{sec:local}
In this Section, we show that the promise problem \sklh\ is $\QMAC$. Let $H_1,\dots,H_m$ an instance of \sklh. We partition the workspace of qubits into disjoints sets $A$ and $B$, each corresponding to $n/2$ qubits. Let $A_i \subset A$ (resp. $B_i \subset B$) the space of qubits in $\spa{A}$ (resp. $\spa{B}$) on which $H_i$ acts. $H_i$ acts on $k$ qubits represented by the space $A_i \otimes B_i$. We can have $A_i = \emptyset$ or $B_i = \emptyset$. Let $H = \sum_i H_i$ (where the summation is over the extension of the $H_i$'s to the entire Hilbert space). The size of the instance is $N = n + m \cdot 2^k$. The term $2^k$ follows from the need of $O(2^k)$ classical bits to describe a $k$-local Hamiltonian.

We use Definition~\ref{def:AlternativeDefinition1} to characterize \sklh\onote{Andre: the previous version was \QMAT. This was a typo, right?}. We are in a $\YES$ instance if there exists $\rho = \rho_A \otimes \rho_B$ such that $\tr{H \rho} \le a$, with $\rho_A \in \spa{A}$ and $\rho_B\in \spa{B}$. We are in a $\NO$ instance if for all $\rho = \rho_A \otimes \rho_B$ with $\rho_A \in \spa{A}$ and $\rho_B\in \spa{B}$, we have $\tr{H \rho} \ge b$

Note that this problem is $\QMAH$ hard. Indeed, if we consider that all the Hamiltonians $H_i$ act only on $\spa{A}$, which means that for all $i$, $B_i = \emptyset$, we obtain an instance of the $k$-local Hamiltonian problem which is $\QMAC$ hence \sklh\ is $\QMAH$. It remains to be shown that \sklh\ is in $\QMA$. 

To show this, we use the fact that another problem, \cldm\ (see Definition~\ref{def:cldm}), is in $\QMA$. More precisely, we consider the \cldm\ problem with $\beta = \frac{b - a}{8m}$. It was shown by Liu~\cite{liu_consistency_2006} that this problem is in $\QMA$. We now describe the $\QMA$ procedure for \sklh. \\
\cadre{ \begin{center}
$\QMA$ protocol for $\sklh$ \end{center}
Let $H_1\dots,H_m$ an instance of $\sklh$. Suppose this is a $\YES$ instance.
Let $\rho = \rho_A \otimes \rho_B$ such that $\tr{H \rho} \le a$
For each Hamiltonian $H_i$, do the following:
\begin{itemize}
\item The prover sends a classical description of the state $\rho_i = \rho^{A_i} \otimes \rho^{B_i}$ where $\rho^{A_i} = \Tr_{A/A_i}(\rho_A)$ and $\rho^{B_i} = \Tr_{B/B_i}(\rho_B)$. This requires sending $O(m\cdot 2^k) = O(N)$ classical bits. 
\item The prover proves to the verifier that the reduced density matrices $\{\rho^{A_i}\}_{i \in [1,m]}$ form a $\YES$ instance of the \cldm\ problem. He also proves that the reduced density matrices $\{\rho_{B_i}\}_{i \in [1,m]}$ form a $\YES$ instance of the \cldm\ problem. 
\item Once the verifier is convinced that the reduced density matrices are consistent, he calculates the value $E = \sum_i \tr{H_i \rho_i}$ and accepts if $E \le a$.
\end{itemize}
}
\subsection{Proof that the protocol works}
\begin{proof}
\textbf{Completeness:} Suppose we are in a $\YES$ instance. This means that there exists $\rho = \rho_A \otimes \rho_B$ such that $\tr{H \rho} \le a$.  The prover sends a classical description of the $\rho_{A_i}$ and $\rho_{B_i}$ where $\rho_{A_i} = \Tr_{A/A_i}(\rho_A)$ and $\rho_{B_i} = \Tr_{B/B_i}(\rho_B)$. Clearly, these reduced density matrices are consistent with $\rho_A$ and $\rho_B$ so the consistency test will pass with probability greater than $2/3$. Then, we have 
\begin{align*}
\tr{H \rho} = \sum_i \tr{H_i \rho} = \sum_i \tr{H_i(\rho_{A_i} \otimes \rho_{B_i})} = E \le a. 
\end{align*}
We conclude that the verifier will accept with probability at least $\frac{2}{3}$.
 
\textbf{Soundness:} Suppose we are in a $\NO$ instance. The prover sends classical descriptions of the states $\rho_{A_i},\rho_{B_i}$. We distinguish two cases:
\begin{itemize}
\item These reduced density matrices fail the consistency test. The verifier accepts with probability smaller than $\frac{1}{3}$.
\item These reduced density matrices pass the consistency test with probability at least $\frac{2}{3}$. This means that there exist two quantum states $\sigma_A,\sigma_B$ such that if we define $\sigma_{A_i} = \Tr_{A/A_i} (\sigma_A)$ and $\sigma_{B_i} = \Tr_{B/B_i} (\sigma_B)$, we have :
\end{itemize}
\begin{align*}
\forall i, \ || \sigma_{A_i} - \rho_{A_i} ||_1 \le \frac{b - a}{8m} \ \mbox{and} \ 
|| \sigma_{B_i} - \rho_{B_i} ||_1 \le \frac{b - a}{8m}.
\end{align*}
Since we are in a $\NO$ instance, for every $\sigma_{A}, \sigma_{B}$ we have 
\begin{align*}
\tr{H (\sigma_A \otimes \sigma_B)} = \sum_i \tr{H_i (\sigma_{A_i} \otimes \sigma_{B_i})} \ge b.
\end{align*}
For each $i$, we have $||(\rho_{A_i} \otimes \rho_{B_i}) - (\sigma_{A_i} \otimes \sigma_{B_i})||_1 \le ||\rho_{A_i} - \sigma_{A_i}||_1 + ||\rho_{B_i} - \sigma_{B_i}||_1 \le \frac{b - a}{4m}$, where the first inequality follows from the subadditivity of the trace distance with respect to tensor products. We now use the following Claim
\begin{claim}[\cite{NC00}]
Let $\rho,\sigma$ two quantum states with $||\rho - \sigma||_1 = \delta$. We have for any positive semidefinite matrix $H \le I, |tr(H(\rho)) - tr(H(\sigma))| \le \delta/2$.
\end{claim} 
Since we do have $H_i$ positive semi definite and $H_i \le I$ for each $i$, we have for all $i$
\begin{align*}
\tr{H_i (\rho_{A_i} \otimes \rho_{B_i})} \ge \tr{H_i (\sigma_{A_i} \otimes \sigma_{B_i})} - \frac{b-a}{2m}.
\end{align*}
Putting this all together, we have
\begin{align*}
E = \sum_{i} \tr{H_i (\rho_{A_i} \otimes \rho_{B_i})} \ge \sum_i \tr{H_i (\sigma_{A_i} \otimes \sigma_{B_i})} - m\cdot \frac{b - a}{2m} \ge b - \frac{a - b}{2} = \frac{a+b}{2},
\end{align*}
therefore, the verifier rejects\onote{if we use the same analysis, the acceptance condition must be $<\frac{a+b}{2}$ and not $\leq \frac{a+b}{2}$}.
\end{proof}
}

\section*{Acknowledgments}
We thank Fernando Brand\~{a}o for his contribution to the soundness proof of Thm.~\ref{thm:ssh_qmatc}.

\bibliographystyle{alpha}
\bibliography{separable_hamiltonian}
\fullversion{
\appendix
\section{Lower bounding the spectral gap}
\label{app:spectral_gap}
In this appendix we prove Claim \ref{cl:gap}. The proof has a similar structure to the one in \cite{KSV02}.
We will first need a few definitions. Given a Hilbert space $\HH$ and a subspace $\LL$, the subspace $\LL^{\perp}$ is the orthogonal complement of the subspace $\LL$ (see e.g. \cite{P91}). Given two subspaces $\LL_1,\LL_2$, the angle $0\leq \theta(\LL_1,\LL_2) \leq \frac{\pi}{2}$ between the subspaces is:
\[ \cos(\theta) \equiv \max_{\ket{\psi_1}\in \LL_1,\ \ket{\psi_2} \in \LL_2} |\braket{\psi_1}{\psi_2}|.\]

Given a Hamiltonian $A \succeq 0$, we define $\Delta(A)$ to be the smallest non-zero eigenvalue of $A$. We use the notation $A \succeq c$ as a shorthand for $A - c I \succeq 0$.
\begin{lemma}[\protect{\cite[Lemma 14.4]{KSV02}}]
\label{le:kitaev_geometric_lemma}
Let $A_1,A_2$ be positive-semidefinite operators, and $\LL_1,\LL_2$ their null subspaces respectively, where $\LL_1 \cap \LL_2 = \{0\}$. Suppose further that $\Delta(A_{1})\geq v$ and $\Delta(A_{2})\geq v$. Then,
\[ A_{1} + A_{2} \succeq v (1-\cos(\theta)), \]
where $\theta$ is the angle between $\LL_{1}$ and $\LL_{2}$.
\end{lemma}
We will use a slightly different version.
\begin{corollary}
\label{co:geometric_lemma}
Let $A_1,A_2$ be positive-semidefinite operators, and $\LL_1,\LL_2$ their null subspaces respectively, where $\LL_1 \cap \LL_2 \equiv \LL$. Suppose further that $\Delta(A_{1})\geq v$ and $\Delta(A_{2})\geq v$. Then,
\[ \Delta(A_{1} + A_{2}) \geq v (1-\cos(\theta)), \]
where $\theta$ is the angle between $\LL_{1} \cap \LL^{\perp}$ and $\LL_{2} \cap \LL^{\perp}$. 
\end{corollary}
The corollary follows from applying Lemma \ref{le:kitaev_geometric_lemma}, to $A_{1}$ and $A_{2}$ with the domain and codomain restricted to $\LL^{\perp}$.


We use Cor.~\ref{co:geometric_lemma} where we substitute $A_{1}= W^{\dagger} H_{in} W$ and $A_{2}= W^{\dagger} H_{prop} W$, where
\[ W= \sum_{t=0}^{T} \ketbra{t}{t} \tensor U_{t}\ldots U_{1}. \]
In this case, the analysis in \cite[Eq. (14.15),(14.16)]{KSV02} shows that $v \geq \frac{c'}{(L+1)^{2}}$ , and we will show that 
\begin{equation}
\label{eq:cos_theta}
\cos^{2}(\theta) \leq 1 - \frac{1}{T+1},
\end{equation}
  which together gives the desired result. 

It can be verified (see  \cite[Eq. (14.13),(14.14)]{KSV02} for a full analysis) that 
\[ \LL_{1} = \ket{0}^{C} \tensor \ket{0^{m}}^{A} \tensor \HH^{P_1,P_2} \bigoplus_{i=1}^{2^{m}-1} \ket{i}^{C} \tensor \HH^{A,P_1,P_2}\]
\[\LL_{2} = \ket{\alpha}^{C} \tensor \HH^{A,P_1,P_2}, \]
where $\ket{\alpha}=\frac{1}{\sqrt{T+1}}\sum_{t=0}^{T} \ket{t}$, and the superscripts denote the subsystems (see Lemma~\ref{le:step_two}). Therefore,
\[ \LL \equiv \LL_{1} \cap \LL_{2} = \ket{\alpha }^{C} \tensor \ket{0^{m}} ^{A}\tensor \HH^{P_1,P_2} \]
\[ \LL_{2} \cap \LL^{\perp} = \ket{\alpha}^{C}\tensor (\bigoplus_{i=1}^{2^{m}-1} \ket{i}^{A} ) \tensor \HH^{P_1,P_2}, \]
 Using these definitions we get:
\begin{align}
\cos^{2}(\theta)&=\max_{\ket{\psi_{1}} \in \LL_{1} \cap \LL^{\perp},\  \ket{\psi_{2}} \in \LL_{2} \cap \LL^{\perp}}  |\braket{\psi_{1}}{\psi_{2}} |^{2} \\
&\leq \max_{\ket{\psi_{1}} \in \LL_{1},\  \ket{\psi_{2}} \in \LL_{2} \cap \LL^{\perp}}  |\braket{\psi_{1}}{\psi_{2}} |^{2} \\
&= \max_{\ket{\psi} \in \LL_{2} \cap \LL^{\perp}}  \bra{\psi} \Pi_{\LL_{1}}\ket{\psi},
\end{align}

where $\Pi_{\LL_1} = \ketbra{0}{0}^{C} \tensor \ketbra{0^{m}}{0^{m}}^{A} \tensor I ^{P} + \sum_{t=1}^{T}\ketbra{t}{t}^{C}\tensor I^{A,P_1,P_2}$ is the projection onto the space $\LL_{1}$. Any state $\ket{\psi} \in \LL_{2}\cap \LL^{\perp}$ can be written in the form $\ket{\alpha}^{C} \tensor \ket{\beta}^{A,P_1,P_2}$, where $\ketbra{0^{m}}{0^{m}}^{A} \tensor I^{P_1,P_2} \ket{\beta} = 0$, therefore we can further bound:
\begin{align*}
\cos^{2}(\theta)&\leq \bra{\alpha} \tensor \bra{\beta}  \ketbra{0}{0}^{C} \tensor \ketbra{0^{m}}{0^{m}}^{A} \tensor I ^{P_1,P_2}  \ket{\alpha} \tensor \ket{\beta} + \bra{\alpha} \tensor \bra{\beta} \sum_{t=1}^{T}\ketbra{t}{t}^{C}\tensor I^{A,P_1,P_2} \ket{\alpha} \tensor \ket{\beta} \\
&= 1 - \frac{1}{T+1},
\end{align*}
where the equality follows from the observation that the first term is $0$ (see the property of $\ket{\beta}$ mentioned above), and that $\bra{\alpha} \sum_{t=1}^{T}\ketbra{t}{t}\ket{\alpha} = \frac{T}{T+1}$. This gives Eq.~\ref{eq:cos_theta}, and completes the proof.
}
\end{document}